\documentclass[12pt]{amsart}%
\usepackage{amsmath}
\usepackage{amssymb}
\usepackage{amsthm}
\usepackage{amscd}
\usepackage{comment}
\usepackage[dvips]{graphicx}
\usepackage[mathscr]{eucal}
\newtheorem{Thm}{Theorem}[section]
\theoremstyle{definition}
\newtheorem{Theorem}[Thm]{Theorem}
\newtheorem{Lemma}[Thm]{Lemma}
\newtheorem{Corollary}[Thm]{Corollary}
\newtheorem{Proposition}[Thm]{Proposition}
\newtheorem{Definition}[Thm]{Definition}
\newtheorem{Example}[Thm]{Example}
\newtheorem{Exercise}{Exercise}
\theoremstyle{remark}
\newtheorem{Remark}{Remark}
\font\sy=cmsy10

\font\ym=msbm10
\newcommand{\Aut}{{\rm Aut}}

\renewcommand{\P}{\text{\ym P}}
\newcommand{\A}{\text{\ym A}}
\newcommand{\B}{\text{\ym B}}

\newcommand{\R}{\text{\ym R}}
\newcommand{\bS}{\text{\ym S}}
\newcommand{\T}{\text{\ym T}}
\newcommand{\C}{\text{\ym C}}

\newcommand{\cE}{{\hbox{\sy E}}}

\newcommand{\cH}{{\hbox{\sy H}}}
\newcommand{\cK}{{\hbox{\sy K}}}
\newcommand{\cL}{{\hbox{\sy L}}}

\newcommand{\sB}{\mathscr B}
\newcommand{\sC}{\mathscr C}
\newcommand{\sE}{\mathscr E}
\newcommand{\sH}{\mathscr H}
\newcommand{\sL}{\mathscr L}
\newcommand{\sQ}{\mathscr Q}
\newcommand{\sS}{\mathscr S}
\newcommand{\sW}{\mathscr W}

\newcommand{\Hom}{\hbox{\rm Hom}}
\newcommand{\End}{\hbox{\rm End}}

\title[]
{Gaussian elements in CCR algebras}
\author[Yamagami Shigeru]{Yamagami Shigeru}

\begin{document}
\maketitle
\begin{center}
Graduate School of Mathematics
\end{center}
\begin{center}
Nagoya University
\end{center}
\begin{center}
Nagoya, 464-8602, JAPAN
\end{center}


\begin{abstract}
A system of matrix units in the Weyl algebra of convolution type
is constructed with the aid of a Gaussian element
so that it includes von Neumann's minimal projection,
which explicitly shows that the associated C*-algebra is a compact operator algebra.
The spectral decomposition of an arbitrary Gaussian element is then worked out
by utilizing the diagonal projections
in the matrix units.
\end{abstract}

\baselineskip=14pt

\section*{Introduction}

The celebrated Stone-von Neumann theorem on the uniqueness of unitary representations of the CCR is sitting
at a branching point of many disciplines and opened a way to a variety of developments both in mathematics
and in physics (\cite{RoJ}). The original proof of von Neumann is based on a Gaussian expression of a certain minimal
projection in the twisted convolution algebra related to the Weyl form of the CCR (\cite{vN}).

Since the way of its proof shows that all representations are unique up to multiplicities,
the associated C*-algebra of convolution type turns out to be a compact operator algebra
thanks to an affirmative answer to Naimark's problem.

We shall here supplement this fact by constructing
matrix units as differentiated Gaussian elements in the convolution Weyl algebra
so that the von Neumann projection are members of them, which enables us to explicitly identify
the convolution type C*-algebra with a compact operator algebra.

The spectral property of Gaussian elements, which is more or less known among specialists
(see \cite[\S~5.5]{Ho} for example),
is also described in terms of these matrix units.

\section{Weyl Relations and Convolution Algebras}
Originally the canonical commutation relations (CCR) were set forth by M.~Born and P.~Jordan
as $[q_j,p_k] = i \delta_{j,k}$, which were afterwards rephrased by H.~Weyl in the form
\[
U_p(x)U_q(y) = e^{ix\cdot y} U_q(y)U_p(x),
\qquad
x,y \in \R^n,
\]
where
$U_p(x) = e^{ix\cdot p}$ and $U_q(x) = e^{ix\cdot q}$ ($x, p \in \R^n$, $x\cdot p = \sum_j x_j p_j$ and so on)
denote $n$-parameter continuous groups of unitaries.
In terms of these unitaries, the Schr\"odinger representation of CCR takes the form
\[
(U_p(y)f)(x) = f(x + y),
\quad
(U_q(y))(x) = e^{ix\cdot y} f(x),
\qquad
f \in L^2(\R^n).
\]
Since unitaries
$U_q(y)$ generate all multiplication operators and any operator in its commutant is realized
again by multiplication, it is reduced to a scalar operator when it furthermore commutes
with the translation operators $U_p(x)$, i.e., the Schr\"odinger representation is irreducible
as is well-known.

Furthermore any irreducible (and continuous) representation of Weyl relations is known to be unitarily equivalent
to the Schr\"odinger representation (the Stone-von Neumann theorem).

For its proof, von Neumann put $U_p$ and $U_q$ together into the form
\[
  U(x,y) = e^{-ix\cdot y/2} U_p(x) U_q(y),
\]
which, being formally equal to $e^{i(xp + yq)}$, is
a continuous family of unitaries parametrized by $(x,y) \in \R^{2n}$ and satisfies
\[
U(x,y) U(x',y') = e^{i(xy' - x'y)/2} U(x+x',y+y').
\]
In other words,
the family $\{ U(x,y)\}$ is a unitary representation of the additive group $\R^{2n}$ twisted by
a two-cocyle $e^{i(xy' - x'y)/2}$.
Thus, relevant is not an euclidean structure but symplectic one governed by $x'y - xy'$.

To make this fact manifest and to simplify the notation at the same time,
regard $\{ p_j,q_j\}$ as a symplectic basis and think of $(x,y) \in \R^{2n}$ as coordinates of
an element $v = \sum (x_j p_j + y_j q_j)$ in a symplectic vector space $V$ with
the symplectic form $\sigma$ given by $\sigma(q_j,p_k) = \delta_{j,k}$ and
$\sigma(p_j,p_k) = 0 = \sigma(q_j,q_k)$.
Notice here that the Liouville measure in $V$ is exactly the Lebesgue measure in $\R^{2n}$.

Associated with such a real symplectic vector space $(V,\sigma)$, we introduce several *-algebras describing
Weyl relations.

Let $\C e^{iV}$ be a free vector space generated by symbols $e^{iv}$ ($v \in V$),
which is a *-algebra by operations
\[
  e^{iv} e^{iw} = e^{-i\sigma(v,w)/2} e^{i(v+w)},
  \quad
  (e^{iv})^* = e^{-iv}
\]
and referred to as a \textbf{Weyl algebra} based on $(V,\sigma)$.
A *-representation $\pi$ of $\C e^{iV}$ is then specified by a family $U(x,y) = \pi(e^{i(xp+yq)})$
of unitaries satisfying \textbf{Weyl relations},
which corresponds to Weyl unitaries exactly when $U(x,y)$ is continuous in parameters $(x,y) \in \R^{2n}$.

Note here that each linear functional $\lambda \in V^*$
gives to a *-automorphism of $\C e^{iV}$ (called \textbf{shift automorphism}) by
$e^{iv} \mapsto e^{if(v)} e^{iv}$ in such a way that it defines an automorphic action of
the additive group $V^*$ on the Weyl algebra $\C e^{iV}$.

We can also work with holomorphically extended objects $e^{v+iw}$ ($v,w \in V$)
which satisfy the obvious *-algebraic operations and span a *-supalgebra $\C e^{V+iV}$ of $\C e^{iV}$.

As an analysis-oriented one, consider the Banach space $L^1(\R^{2n}) = L^1(V)$ relative to the Liouville measure,
which is made into a Banach *-algebra
(denoted by $L^1(V,\sigma)$ and referred to as
a \textbf{convolution Weyl algebra})
so that
\[
  f \mapsto  \pi(f) = \int_V f(v) \pi(e^{iv})\, dv
  = \int_{\R^{2n}} f(xp + yq) U(x,y)\, dxdy
\]
gives a *-representation of $L^1(V,\sigma)$: For $f, g \in L^1(V)$,
\[
  (fg)(v) = \int_V e^{i\sigma(v,v')/2} f(v') g(v-v')\, dv',
  \quad
f^*(v) = \overline{f(-v)}.
\]
The automorphic action of $V^*$ on $\C e^{iV}$ by shifts is also converted to that on the convolution algebra by
$f(v) \mapsto e^{i\lambda(v)} f(v)$.

Notice that, as in the case of group algebras of locally compact groups,
a continuous representation $\pi(e^{iv})$ of $\C e^{iV}$ 
is in one-to-one correspondence with a *-representation $\pi(f)$ of $f \in L^1(V,\sigma)$.



\begin{Remark}
Thanks to the inequality $\| fg\|_2 \leq \| f\|_1 \| g\|_2$ for $f \in L^1(V)$ and $g \in L^2(V)$,
$L^1(V) \cap L^2(V)$ is a *-subalgebra of $L^1(V,\sigma)$,
which turns out to be a (unimodular) Hilbert algebra with respect to the $L^2$-inner product.
The associated trace $\tau$ is therefore described by
$\tau(f^*f) = (f|f)$ ($f \in L^1 \cap L^2$).
\end{Remark}




Although $e^{iv}$ itself is not in $L^1(V,\sigma)$,
a formal multiplication of $e^{iv}$ on $\int f(v') e^{iv'}\, dv'$
enables us to realize $e^{iv}$ as a multiplier of $L^1(V,\sigma)$:
\[
  (e^{iv}f)(v') = e^{-i\sigma(v,v')/2} f(v'-v),
  \quad
  (fe^{iv})(v') = e^{i\sigma(v,v')/2} f(v'-v).
\]
This multiplier realization is compatible with the shift automorphisms
on $\C e^{iV}$ and $L^1(V,\sigma)$.

To get an analogous realization of $\C e^{V+iV}$ as a multiplier algebra, we further introduce
a dense *-subalgebra $L^1_\varpi(V,\sigma)$ of $L^1(V,\sigma)$ consisting of entire functions:
By definition, a function $f \in L^1(V)$ belongs to $L^1_\varpi(V,\sigma)$ if
$f(v)$ is continuous in $v \in V$ and extended to an entirely analytic function
$f(v+iw)$ of $v + iw \in V^\C$ so that for each $\lambda \in V^*$
$e^{\lambda(v)} f(v+iw)$ is in $L^1(V)$ as a function of $v \in V$ and depends norm-continuously
on $w \in V$.

Given a function
$f \in L^1(V)$ of supexponential decay
(i.e., $f(v) = O(e^{-r|v|})$ for any $r>0$), its Gaussian regularization belongs to $L^1_\varpi(V)$.
In fact, given an inner product $\langle\cdot,\cdot\rangle$ in $V$, the Gaussian regularization
\[
  \phi(v) = \int_V e^{-\langle v',v'\rangle} f(v-v')\, dv'
\]
of $f$ is holomorphically extended and the expression
\begin{align*}
  \phi(v+iw) e^{\lambda(v)}
&= \int_V e^{-\langle v-v',v-v'\rangle - 2i\langle w,v-v'\rangle + \langle w,w\rangle + \lambda(v-v')}
                                                                   f(v') e^{\lambda(v')}\, dv'\\
&= \int_V e^{-\langle v-v'+u+iw,v-v'+u+iw\rangle + \langle u,u\rangle + 2i\langle u,w\rangle}
                                                                   f(v') e^{\lambda(v')}\, dv'
\end{align*}
with $u \in V$ defined by $\lambda(\cdot) = 2\langle u, \cdot \rangle$
shows that
\[
  \phi(v+iw) e^{\lambda(v)} = \varphi(v+u+iw) e^{\langle u,u\rangle + 2i\langle u,w\rangle}
\]
is norm-continuous in $w \in V$ as an $L^1(V)$-valued function,
where $\varphi$ is the Gaussian regularization of $e^\lambda f$.


Now for $f, g \in L^1_\varpi(V,\sigma)$, their product in $L^1(V,\sigma)$ is holomorphically extended to
\[
  (fg)(v+iw) = \int_V e^{i\sigma(v+iw,v')/2} f(v') g(v+iw-v')\, dv'.
\]
Note here that $e^{i\sigma(v+iw,v')} f(v')$ is integrable as a function of $v'$ which depends on $v+iw$ holomorphically.
Moreover, for $\lambda \in V^*$,
\[
  (fg)(v+iw) e^{\lambda(v)}
  = \int_V e^{\lambda(v') + \sigma(v',w)/2}f(v')\, e^{\lambda(v-v')}g(v+iw-v') e^{i\sigma(v,v')/2}\, dv'
\]
is continuous in $w \in V$ as an $L^1(V)$-valued function because this is
a product of $e^{\lambda(v) + \sigma(v,w)/2} f(v)$ and $g^{\lambda(v)} g(v+iw)$ in $L^1(V,\sigma)$ and
these are norm-continuous in $w \in V$ as $L^1(V)$-valued functions.

Finally, $e^{\lambda(v)} f^*(v+iw) = e^{\lambda(v)} \overline{f(-v + iw)}$
is holomorphic and continuous in $w \in V$ as an $L^1(V)$-valued function.

Here is a summary so far.

\begin{Lemma}
$L^1_\varpi(V,\sigma)$ is a dense *-subalgebra of $L^1(V,\sigma)$.
\end{Lemma}

For $v+iw \in V + iV$, a formal identity
\[
  e^{iv-w} \int_V f(v') e^{iv'}\, dv' = \int_V f(v'-v-iw) e^{-i\sigma(v+iw,v')/2} e^{iv'}\, dv'
\]
and a similar expression for right multiplication suggest putting
\begin{align*}
  (e^{iv-w}f)(v') &= e^{-i\sigma(v+iw,v')/2} f(v'-v - iw),\\
  (fe^{iv-w})(v') &= e^{i\sigma(v+iw,v')/2} f(v'-v - iw).
\end{align*}

Now the following is immediate to check.

\begin{Lemma}
The *-algebra $\C e^{V+iV}$ is realized as a multiplier algebra of $L^1_\varpi(V,\sigma)$.
\end{Lemma}

\section{von Neumann's Projection and Matrix Units}
A hermitian element in $L^1(V,\sigma)$ is said to be \textbf{Gaussian} if it is of the form
$e^{-(v|v) + i\lambda(v) + \mu}$, where $\mu \in \R$, $\lambda \in V^*$ and $(v|w)$ is a real inner product in $V$.
Clearly Gaussian elements belong to $L^1_\varpi(V,\sigma)$ with their linear parts $\lambda(v)$ realized
by the effect of the shift automorphism associated to $\lambda \in V^*$.

Given an inner product $(v|w)$ in a finite-dimensional symplectic vector space $(V,\sigma)$,
thanks to the standard form of symplectic matrices,
we can find an orthonormal basis $\{ e_j,f_j\}$ of $V$ so that $\sigma_j = \sigma(f_j,e_j) > 0$
and $\{ \R e_j + \R f_j\}$ is a $\sigma$-orthogonal family.
Then $\{ p_j = \sigma_j^{-1/2} e_j, q_j = \sigma_j^{-1/2}f_j\}$
is a canonical basis which also diagonalizes the inner product so that
$(p_j|p_j) = (q_j|q_j) = 1/\sigma_j$.
In this way, Gaussian elements are factored into two-dimensional ones and their spectral properties are
more or less reduced to the case of single freedom.


With this observation in mind, we now
introduce an element $g_\alpha(x,y) = e^{-(x^2+y^2)/4\alpha}$
($\alpha \in \C$, $\text{Re}\,\alpha>0$) in $L^1(V,\sigma)$
for the canonical choice $V = \R^{2n}$ ($g_\alpha$ being Gaussian for $\alpha > 0$)
and investigate its spectral properties in the convolution Weyl algebra $L^1(V,\sigma)$
or in its C*-envelope $C^*(V,\sigma)$. Notice that in terms of the canonical coordinates,
the convolution product is expressed by
\[
  (fh)(x,y) = \int_{\R^{2n}} e^{i(x'y-xy')/2} f(x',y') h(x-x',y-y')\, dx' dy'.
\]

By a simple calculation based on Gaussian integrals, we see that
\[
(g_\alpha g_\beta)(x,y) = \left( \frac{4\pi\alpha\beta}{\alpha + \beta} \right)^n
\exp\left(-\frac{x^2 + y^2}{4\gamma}
\right),
\quad
\gamma = \frac{\alpha + \beta}{1 + \alpha\beta}.
\]
Since each $\gamma \leq 1$ is of the form $\gamma = 2/(\alpha + \alpha^{-1})$
for the choice $\beta^{-1} = \alpha > 0$,
the formula implies that $g_\gamma = g_\alpha^2 = g_{1/\alpha}^2 \geq 0$.
Consequently $g_\gamma^{1/2} = g_\alpha \geq 0$ for $\alpha \leq 1$, whereas
$g_{1/\alpha} \not= g_\alpha$ $(\alpha < 1$) is a non-positive root of $g_\gamma$.
In particular, when $\alpha = \beta = 1$,
\[
g(x,y) = \frac{1}{(2\pi)^n} e^{-(x^2+y^2)/4}
\]
is a projection in $L^1(V,\sigma) \subset C^*(V,\sigma)$. 

In von Neumann's proof of uniqueness of Weyl unitaries,
it is a key to observe that the projection $g \in L^1(V,\sigma)$ is minimal in the sense that
$ge^{i(xp + yq)}g = e^{-(x^2+y^2)/4} g$ ($x,y \in \R^n$).

We shall now construct a system of
matrix units in $L^1(V,\sigma)$ so that it includes the projection $g$ as a diagonal member.
To elucidate the role played by $g$ in finding matrix units inside $L^1(V,\sigma)$, introduce complex parameters
$z = (x+iy)/\sqrt{2}$, $w = (x'+iy')/\sqrt{2}$ and rewrite
\begin{multline*}
(e^{i(xp+yq)} g e^{i(x'p+y'q)})(s,t)\\
= e^{i(s(y'-y) - t(x'-x))/2} e^{i(x'y-xy')/2} g(s-x-x',t-y-y')\\
= \frac{1}{(2\pi)^n} e^{-((s-x-x')^2 + (t-y-y')^2)/4} e^{i(s(y'-y) - t(x'-x))/2} e^{i(x'y-xy')/2}
\end{multline*}
to get the expression
\[
\frac{1}{(2\pi)^n} e^{-(s^2+t^2)/4} e^{-(|z|^2 + |w|^2)/2}
\exp\left( \frac{s+it}{\sqrt{2}} \overline{z}
+ \frac{s-it}{\sqrt{2}} w - \overline{z}w \right).
\]
In view of $i(xp + yq) = za - \overline{z}a^*$ and the identity
$e^{i(xp + yq)} = e^{-|z|^2/2} e^{-\overline{z}a^*} e^{za}$,
with annihilators $a = (q+ip)/\sqrt{2}$ and creators
$a^* = (q - ip)/\sqrt{2}$ satisfying $[a_j,a_k^*] = \delta_{j,k}$,
\[
e^{-\overline{z}a^*} e^{za}g e^{-\overline{w}a^*} e^{wa}
= e^{(|z|^2+|w|^2)/2} e^{i(xp + yq)} g e^{i(x'p + y'q)}
\]
is expressed by a function
\[
\frac{1}{(2\pi)^n}
\exp\left( -\frac{s^2+t^2}{4} + \frac{s+it}{\sqrt{2}} \overline{z}
  + \frac{s-it}{\sqrt{2}} w - \overline{z}w \right)
\]
of $(s,t) \in \R^{2n}$ in $L^1_\varpi(V,\sigma)$.

Here we pay attention to the analytic dependence on parameters $z$, $w$.
Clearly the last function is antiholomorphic in $z$ and holomorphic in $w$, whence
the same behavior of $e^{-\overline{z}a^*} e^{za}g e^{-\overline{w}a^*} e^{wa}$
reveals that $e^{za}g e^{-\overline{w} a^*} = g$ and 
\[
(e^{-\overline{z}a^*} g e^{wa})(s,t)
= \frac{1}{(2\pi)^n}
\exp\left( -\frac{s^2+t^2}{4} + \frac{s+it}{\sqrt{2}} \overline{z}
+ \frac{s-it}{\sqrt{2}} w - \overline{z}w \right).
\]
Now the minimality can be read off from these: Thanks to the Weyl relations and $g^2 = g$,
\[
e^{|z|^2/2} ge^{za} e^{-\overline{z} a^*} g = g e^{za - \overline{z} a^*} g
= e^{-|z|^2/2} g e^{- \overline{z} a^*} e^{za} g
= e^{-|z|^2/2} g,
\]
whence $ge^{i(xp + yq)} g = e^{-(x^2 + y^2)/4} g$.

Moreover, from the fact that the coefficient in the right hand side is equal to the evaluation of
$e^{i(xp + yq)}$ by the Fock state $\omega$, one sees that
$C^*(V,\sigma) g C^*(V,\sigma)$ is *-isomorphic to a *-subalgebra
\[
  \pi(C^*(V,\sigma))|\omega^{1/2})(\omega^{1/2}| \pi(C^*(V,\sigma))
\]
of $\sB(\overline{C^*(V,\sigma)\omega^{1/2}})$ via
the correspondence $fgh \leftrightarrow \pi(f)|\omega^{1/2}) (\omega^{1/2}| \pi(h)$ ($f,h \in L^1(V,\sigma)$),
where $\pi$ denotes the standard Fock representation of $C^*(V,\sigma)$
with $\omega^{1/2}$ the Fock vacuum vector.

Since $g_{z,w} = e^{(|z|^2 + |w|^2)/2} e^{-\overline{z} a + za^*}g e^{w a  - \overline{w}a^*} = e^{za^*} g e^{wa}$
is holomorphic as an $L^1(\R^{2n})$-valued function of $z,w \in \C^n$,
it allows a Taylor expansion\footnote{Formally $g_{k,l} = \frac{1}{\sqrt{k!l!}} (a^*)^k g a^l$ and it corresponds to
  $\frac{1}{\sqrt{k!l!}}(a^*)^k |\omega^{1/2})(\omega^{1/2}| a^l$ in the Fock representation.}
\[
  g_{z,w} = \sum_{k,l \geq 0} \frac{1}{\sqrt{k!l!}} z^kw^l g_{k,l},
  \quad
  g_{k,l} \in L^1(\R^{2n}).
\]
We claim that $(g_{k,l})_{k,l \geq 0}$ constitute matrix units in $L^1(V,\sigma)$.
In fact, $g_{\overline{w},\overline{z}}^* = g_{z,w}$ and $g_{z,w} g_{z',w'} = e^{wz'} g_{z,w'}$
give $g_{k,j}^* = g_{j,k}$ and $g_{j,k} g_{l,m} = \delta_{k,l} g_{j,m}$ respectively.

Finally we show that $\{ g_{z,w}; z, w \in \C^n\}$ is total in $L^1(\R^{2n})$.
To see this, let $f \in L^\infty(\R^{2n}) = L^1(\R^{2n})^*$ satisfy
\begin{multline*}
0 = \int_{\R^{2n}} f(s,t) g_{z,w}(s,t)\, dsdt\\
= \frac{e^{zw}}{(2\pi)^n} \int f(s,t) e^{-(s^2+t^2)/4 -(s+it)z/\sqrt{2} + (s-it)w/\sqrt{2}}\, dsdt
\end{multline*}
for any $z, w$. Then the Fourier transform of $f(s,t)e^{-(s^2+t^2)/4}$ vanishes and hence $f = 0$.

\begin{Theorem}
  The C*-envelope $C^*(V,\sigma)$ of $L^1(V,\sigma)$ is a compact operator algebra
  generated by matrix units $\{ g_{k,l} \}$ and
  continuous Weyl unitaries are unitarily equivalent to an ampliation of the Fock representation.

  Moreover the canonical trace is given by the formula
  \[
    \text{tr}(f^*f) = (2\pi)^n \int_V |f(v)|^2\, dv
  \]
  for $f \in L^1(V) \cap L^2(V)$.
  Note that $(f^*f)(v) = (f|fe^{-iv})$ is a continuous function of $v \in V$ and
  $\text{tr}(f^*f) = (2\pi)^n (f^*f)(0)$.
\end{Theorem}

\begin{Corollary}[Stone-von Neumann]　
  Irreducible Weyl unitaries are unique up to unitary equivalence.　
\end{Corollary}

\begin{Remark}~
  {\small
    \begin{enumerate}
      \item
  Since $C^*(V,\sigma)$ is separable, the theorem itself follows from the corollary
  as an affirmative case of Naimark's problem \cite{RoA}.
\item
  An explicit form of $g_{k,l}$ is dealt with in \cite{Ha}
  in the context of Fock representation.
\item
  The fact that $C^*(V,\sigma)$ is a compact operator algebras also follows from
  $C_0(G) \rtimes G \cong \sC(L^2(G))$ for a locally compact group $G$,
  where  the corssed product is taken with respect to the translational automorphic action of $G$ on
  $C_0(G)$ and $\sC(L^2(G))$ is the compact operator algebra on $L^2(G)$,
  see \cite{Ta} for more information. 
\end{enumerate}
}
\end{Remark}

\section{Spectral Analysis of  Gaussian Elements}
We here restrict ourselves to the case $V = \R^2$ for the time being and
work out the spectral decomposition of $g_\gamma$ ($\text{Re}\,\gamma > 0$)
in the C*-algebra $C^*(V,\sigma)$.
Since $e^{-\overline{z}a^*} g e^{za} \in L^1_\varpi(V,\sigma)$ ($z = (x+iy)/\sqrt{2}$) is realized by the function
\[
  \frac{1}{2\pi}
  \exp\left( - \frac{s^2+t^2}{4} + sx + ty - \frac{x^2+y^2}{2} \right),
\]
of $(s,t) \in \R^2 = V$, its parameter dependence on $z \in \C$ is $L^1$-continuous and
$\| e^{-\overline{z}a^*} g e^{za} \|_1 = 2 e^{(x^2+y^2)/2}$.
Thus, if $\text{Re}\,\mu > 1$,
$e^{-\mu |z|^2} e^{-\overline{z}a^*} g e^{za} \in L^1(V,\sigma)$ is norm-integrable with respect to $z$
and we see that
\[
  g_\gamma = \frac{2\gamma}{\gamma -1} \int_{\R^2}
  e^{-\mu|z|^2} e^{-\overline{z}a^*} g e^{za}\, dxdy
\]
for the choice $\mu = (\gamma + 1)/(\gamma - 1)$
with $\text{Re}\, \mu > 1 \iff \text{Re}\,\gamma > 1$.

We now think of this integral expression in $C^*(V,\sigma)$,
where the Taylor expansion
\[
  e^{-\overline{z}a^*} g e^{za}
  = \sum_{k,l \geq 0} \frac{(-\overline{z})^k z^l}{\sqrt{k!l!}} g_{k,l}
\]
is absolutely convergent in view of the estimate
\[
  \|   e^{-\overline{z}a^*} g e^{za} \|
  \leq \sum_{k,l \geq 0} \frac{|z|^{k+l}}{\sqrt{k!l!}},
\]
and the integration is then calculated termwise to get the expression
\begin{align*}
  g_\gamma &= \frac{2\gamma}{\gamma -1}
      \sum_{k,l \geq 0} \frac{1}{\sqrt{k!l!}} \int e^{-\mu|z|^2}
              (-\overline{z})^k z^l g_{k,l}\, dxdy\\
           &= 2\pi \frac{2\gamma}{\gamma -1}
             \sum_{k \geq 0} \frac{(-1)^k}{k!} \int_0^\infty
             e^{- \mu r^2/2} r^{2k+1}\, dr\ g_{k,k}\\
           &= 2\pi \frac{2\gamma}{\gamma -1} \sum_{k=0}^\infty \frac{(-1)^k}{\mu^{k+1}} g_{k,k}
           = 2\pi \frac{2\gamma}{1+\gamma} \sum_{k=0}^\infty
             \left(\frac{1-\gamma}{1+\gamma}\right)^k g_{k,k}
\end{align*}
for $\text{Re}\, \mu > 1 \iff \text{Re}\, \gamma > 1$.
Since both $g_\gamma$ and the last expression are holomorphic on $\text{Re}\,\gamma\, >0$
as elements in $C^*(V,\sigma)$,
the above equality holds in the region $\text{Re}\, \gamma > 0$ as well.

Notice that $g_\gamma$ ($\text{Re}\,\gamma > 0$) is a normal element of trace class with
\[
  \text{tr}(g_\gamma) = 2\pi \frac{2\gamma}{1+\gamma}
  \sum_{k=0}^\infty \left( \frac{1-\gamma}{1+\gamma} \right)^k = 2\pi.
\]

\begin{Remark}
{\small
  In the Schr\"odinger representation, the above spectral decomposition of $g_\gamma$ is reduced to
  Mehler's formula in \cite[\S 1.7]{Fo}.
}
\end{Remark}

Now consider the general case $V = \R^{2n}$ and
let $\{ \gamma_j\}_{1 \leq j \leq n}$ be positive eigenvalues of $i\sigma/4$ including multiplicity
with respect to the quadratic part of a Gaussian element.
Then the spectrum of $\exp(-\sum_j (x_j^2+y_j^2)/4\gamma_j)$ in $C^*(V,\sigma)$ is
\[
  (2\pi)^n \prod_j \frac{2\gamma_j}{1+\gamma_j} \left\{
\prod_j \left( \frac{1-\gamma_j}{1+\gamma_j} \right)^{l_j}; l_j \geq 0 \right\}.
\]

In particular, the Gaussian element is positive if and only if $\gamma_j \leq 1$ for $1 \leq j \leq n$.
If this is the case, the $r$-th power ($r>0$) of $g_\gamma$ with $\gamma = \tanh\theta$ ($\theta>0$)
is again Gaussian and given by
\[
  (4\pi)^{r-1} \frac{(\sinh \theta)^r}{\sinh(r\theta)}
  \exp\left( - \frac{x^2+y^2}{4\tanh(r\theta)} \right)
\]
as seen in \cite{SF}.

As a final remark, we describe the following known fact
(see \cite[Theorem~5.5.1]{Ho} for example) in our context.

\begin{Proposition}
  There is one-to-one correspondence between free states of $C^*(V,\sigma)$
  and positive normalized Gaussian elements without scalar shifts.
\end{Proposition}

\begin{proof}
  We assume notations and results on free states in \cite{gqfs}.
  Through the density operator realization, free states are given by positive normalized Gaussian elements.

  Conversely, a positive Gaussian element is factored into single-freedom ones and we may assume that $V = \R^2$.
Define a sesquilinear form $S$ on $V^\C$ by
\[
S + \overline{S} =
\begin{pmatrix}
  \frac{1}{\gamma} & 0\\
  0 & \frac{1}{\gamma}
\end{pmatrix},
\quad
S - \overline{S} = i\sigma =
\begin{pmatrix}
  0 & -i\\
  i & 0
\end{pmatrix}.
\]
Then
\[
S = \frac{1}{2}
\begin{pmatrix}
  1/\gamma & -i\\
  i & 1/\gamma
\end{pmatrix} \geq 0
\iff 0 < \gamma \leq 1
\]
shows that the Gaussian element $g_\gamma$ gives the density operator of a free state.
\end{proof}

\end{document}